\documentclass[pdflatex,sn-mathphys-num]{sn-jnl}

\usepackage{graphicx}%
\usepackage{multirow}%
\usepackage{amsmath,amssymb,amsfonts}%
\usepackage{amsthm}%
\usepackage{mathrsfs}%
\usepackage[title]{appendix}%
\usepackage{xcolor}%
\usepackage{textcomp}%
\usepackage{manyfoot}%
\usepackage{booktabs}%
\usepackage{algorithm}%
\usepackage{algorithmicx}%
\usepackage{algpseudocode}%
\usepackage{listings}%
\usepackage{braket}
\usepackage{color, soul}

\theoremstyle{thmstyleone}%
\newtheorem{theorem}{Theorem}
\theoremstyle{thmstyletwo}%
\newtheorem{remark}{Remark}%
\newtheorem{corollary}{\textbf{Corollary}}
\DeclareMathOperator*{\argmax}{arg\ max}

\theoremstyle{thmstylethree}%

\begin{document}

\title[Optimal phase change for Grover's algorithm]{Optimal phase change for a generalized Grover's algorithm}

\author[1]{\fnm{Christopher} \sur{Cardullo}}\email{cacardul@ncsu.edu}

\author[1]{\fnm{Min} \sur{Kang}}\email{mkang2@ncsu.edu}

\affil[1]{\orgdiv{Department of Mathematics}, \orgname{North Carolina State University}, \orgaddress{\street{2108 SAS Hall}, \city{Raleigh}, \postcode{27695}, \state{North Carolina}, \country{United States of America}}}

\abstract{We study the generalized Grover's algorithm to find the optimal phase change for maximizing the gain in probability for the target state given an arbitrary amplitude vector. In the classic setting of Grover's algorithm with a real initial amplitude vector, we rigorously identify three different regions for the target state in which the corresponding optimal phase change differs significantly. Careful study into these regions reaffirms that a greedy algorithm approach to Grover's would result in a final probability that is less than one. In a setting where the amplitude vector is complex, we find that the optimal phase change depends non-trivially on the phase of the complex amplitude of the target state. To help with more intuitive understanding, we show the concrete surface of optimization for a set of size $N=2^{10}$ as well as provide a heuristic explaining how the phase of the target's amplitude influences the optimal phase change when the amplitude stays away from 0 and 1.}



\maketitle

\section{Introduction}\label{sec1}

Search algorithms are an important tool in mathematics, computer science, and other related areas, as they allow for a space to be searched for either one or multiple target elements that match a given criteria. Applications of this include optimization, root finding, or any problem where verification of a solution is easy but finding a solution is difficult. Two broad categories we can break search algorithms into are structured search algorithms, where we have some prior knowledge of the underlying data set, and unstructured search algorithms, where we do not. Unstructured search is of particular interest in the field of quantum computation because it shows a clear advantage that quantum computers provide. Proved by Grover in 1996 \cite{Gro1}, by exploiting quantum superposition and entanglement, the time required to perform an unstructured search can be brought down from $\mathscr{O}(N)$ to $\mathscr{O}(\sqrt{N})$ for a system of size $N$. Grover's algorithm is an example of an amplitude amplification method, and it has been shown to have the optimal complexity level a quantum unstructured search can achieve, that of $\mathscr{O}(\sqrt{N})$\cite{BBBV}. The primary mechanism underlying Grover's algorithm is the use of Housholder reflections to amplify the amplitude of the target state while simultaneously suppressing the amplitudes of non-target states.

We begin by reviewing the development of Grover's algorithm and the major research directions that have emerged since its introduction. Early on, research into arbitrary initial amplitudes, both pure and mixed states, was done to show the robustness of Grover's algorithm \cite{BBBGL}\cite{BK}. In addition, researchers have investigated general phase changes to better understand Grover's algorithm and its generalizations \cite{Sz}. In recent years, Grover's algorithm has seen use in quantum machine learning and neural networks tasks \cite{HRN}\cite{NSLWLL}. It has also been used as a subroutine in other quantum algorithms \cite{GGZW}\cite{GWG}.

An active area of research that is particularly relevant to our work is deterministic Grover's algorithm which aims to modify the original Grover's algorithm so that the target probability reaches exactly one. Several approaches have been explored to accomplish this objective. One such method involves determining a phase change that differs from $\pi$, and creating a new Grover iterate that rotates by that phase change at each step. This method can be used by enforcing the two rotations performed to use the same phase change \cite{Lo1}\cite{ZZDW}, or to allow the rotations to utilize different phase changes as long as they satisfy a determined relationship \cite{Ho1}. More recently, the approach of allowing for two different rotations that satisfy the relationship $\phi = -\psi$ was shown to be optimal in terms of iteration count \cite{GTP}. In recent literature, the term phase matching is used to describe the relationship between the two phase changes, and not that they are necessarily equal. This is just one method we chose to highlight in the area of deterministic Grover's algorithm, however there are many more methods that have been developed to achieve deterministic, or near-deterministic results \cite{LLparam}\cite{MBR}\cite{WZK}\cite{GTP2}\cite{RJS}. In many of these approaches, the number of steps required is only one more than the classical Grover's algorithm.

Our research focuses on two major questions concerning generalized Grover's algorithm forcing that the two phase changes within are identical. First, we seek to rigorously establish that a greedy optimization strategy in this generalization does not outperform the original Grover's algorithm with Hadamard's initial condition, and to characterize the conditions under which it fails. Second, we investigate how if a Hadamard initialization isn't used and the amplitude vector is allowed to be a general complex vector, the optimal phase change differs from $\pi$. To address the first question, we identify in Theorem \ref{th:RealVec} three distinct regions for the amplitude vector, together with the corresponding optimal phase changes. We show that there exists a small region with size in the order of $N^{-3/2}$ where the optimal phase change differs from $\pi$, and we provide a closed form formula for the optimal phase. Before reaching this region, for the majority of potential amplitudes, the optimal phase change remains $\pi$, supporting the conventional understanding of locally optimized Grover's algorithm. After this region, the optimal phase change drops to zero. Numerical results illustrating the regions at each iteration are presented in Tables \ref{tab:25general} and \ref{tab:210general}.

For the second question, we work with a general complex amplitude vector. We use analytical methods to determine the optimal amplitude after one iteration, then use numerical methods to identify the surface of optimization. This enables us to observe that the optimal phase change exhibits a non-trivial dependence on the phase of the target state in the amplitude vector. This relationship is illustrated in Figure \ref{fig:1}, and the optimization statement is given in Corollary \ref{cor:argmax}.

Lastly, we want to provide a brief description on how the rest of the paper is structured, to help the readers. In section 2, we provide necessary preliminaries on Grover's algorithm, and Long's generalization. In section 3, we derive the optimization formula for the probability of our target state given an arbitrary amplitude vector, found in \eqref{eq:Optimizaiton}. We discuss the relationship we find here in Figures \ref{fig:1} and \ref{fig:rough_change}. In section 4 we apply this approach to the regular Grover's algorithm setting with an initialization performed by the Hadamard gate and identify three distinct regions where the optimal value of $\phi$ changes completely. And in section 5, we discuss the conclusion and our further research interests and directions.

\section{Preliminaries}\label{sec2}
Consider the space $\mathbf{S}$ with $|\textbf{S}| = 2^n = N$ which contains all elements we are searching over. We will assign each element in $\textbf{S}$ a unique number between $0$ and $N-1$ (ordering does not matter). Let $\tau \in \textbf{S}$ be the target element we are searching for. Consider the standard basis for $\mathbb{C}^2$, renamed as $\textbf{e}_1 = \ket{0}$ and $\textbf{e}_2 = \ket{1}$. $\ket{\cdot}$ and $\bra{\cdot}$ are called the Dirac Bra-Ket notation, with $\ket{\cdot}$ being a column vector and $\bra{\cdot}$ being a row vector intuitively. We also can use the Dirac Bra-Ket notation in discussion of binary representation. Take for example $\ket{5}_3$, with the subscript 3 standing for the number of binary places we are working with. So we have that $\ket{5}_3 = \ket{101}_3 = \ket{1}\otimes \ket{0}\otimes \ket{1} = [0\ 0\ 0\ 0\ 0\ 1\ 0\ 0]^\top$.

Consider the space $\mathbb{C}^{2\otimes n}$. Its basis will be $\{\ket{x}_n\}$ with $x\in\{0,1,\dots,N-1\}$ and written in binary format. We will consider the subspace $\textbf{E}_n\subset\mathbb{C}^{2\otimes n}$ s.t. $\sum_{0\leq x < N} a_x\ket{x}_n = \textbf{a}\in \textbf{E}_n$ if $\sum_{0\leq x<N}|a_x|^2 = 1$. In the context of quantum mechanics, our n-qubits are in $\textbf{E}_n$, and we define $a_x$ to be the amplitude associated with $\ket{x}_n$ and $|a_x|^2$ is the probability that upon observation, our n-qubit will collapse to $\ket{x}_n$. Consider $\textbf{a}\in \textbf{E}_2$. Then $\textbf{a} = a_0\ket{00}_2 + a_1\ket{01}_2+a_2\ket{10}_2+a_3\ket{11}_2 =$ $[a_0\ a_1\ a_2\ a_3]^\top$. Thus we can translate an element from $\textbf{E}_n\subset\mathbb{C}^{2\otimes n}$ to one in $\mathbb{C}^{2^n}$.

The last crucial aspect of quantum computing we need to cover are gates and how we do operations on states. Firstly, we have that all gates can be represented as matrices that act on states, which as we showed just before, can be represented as column vectors in $\mathbb{C}^{2^n}$. Next, there can be no information loss in the system, meaning that all gates must be reversible, and hence the matrices must be unitary.

Now we will have a more in-depth discussion of Grover's algorithm, and a generalized Grover's algorithm proposed by Long. Grover's algorithm is comprised of the following steps \cite{Gro1},

\begin{enumerate}
    \item Use the Hadamard gate to equalize the amplitude for all qubits
    \item Rotate the amplitude of the target element $\ket{\tau}$ by $\pi$ radians. This process is represented by the operator $I^{\pi}_{\ket{\tau}}$. 
    \item Apply the product of operators $FI^{\pi}_{\ket{0}}F$. $F$ is  the Walsh-Hadamard transform operator, and is defined as $2^{-n/2}(-1)^{\bar{i}\cdot \bar{j}}$ with $n$ being the number of qubits in the system, $\bar{i}\cdot \bar{j}$ being the binary dot product of $i$ and $j$, and our index starting at 0. 
\end{enumerate}

Steps 2 and 3 are repeated until the probability of the target $\ket{\tau}$ identified as $P(\ket{\tau})$ is sufficiently large. We call this the Grover iterate, and write it as $G \equiv -FI^{\pi}_{\ket{0}}FI^{\pi}_{\ket{\tau}}$. It can be shown that the optimal number of iterations with the classical Grover's algorithm is $\lfloor\frac{\pi}{4}\sqrt{N}\rfloor$. Note, that if we were to do more than this, the probability for the target would in fact start to decrease, so it is important that we know the optimal number of times the Grover iterate is applied. In fact, Grover also showed that this search can be performed with any unitary gate, not just the Walsh-Hadamard transform operator \cite{Gro2}.

Building off of this work, Long generalized this by allowing for a phase change that is not strictly $\pi$ \cite{Lo1}. His motivation for this was from the case that the system is relatively small, hence accuracy becomes most important. Long found that the phase change of $\pi$ is not optimal in this case. The way he implemented these arbitrary phase rotations is by using the following operator which generalizes classical Grover's algorithm with arbitrary phase $\phi$.

\[I^{\phi}_{\ket{\tau}} = I+(e^{i\phi}-1)\ket{\tau}\bra{\tau}.\]

We can see here that if $\phi = \pi$, then we get a so-called Householder reflection, which is the exact reflections used in the classical Grover's algorithm. To simplify the computation further, we reduce the iteration into 2-dimensional linear maps by adopting the approach used by Long. We condense the non-target elements into a single basis vector $\ket{a}$, observing that they all go through the same change in amplitude in each iteration. Note that it is properly normalized so that the existing computations get simpler. This new basis vector $\ket{a}$ is defined as

\[\ket{a} = \frac{1}{\sqrt{N-1}}\sum_{i\neq \tau} \ket{i}.\]

This condenses our amplitude vector into a vector of $\mathbb{C}^2$. We can now calculate the amplitude after one iteration using the linear map identified by Long \cite{Lo1}

\begin{equation}
    \begin{bmatrix}
    -e^{i\phi}(1+(e^{i\phi}-1)\sin^2\beta) & -(e^{i\phi}-1)\sin\beta \cos\beta\\
    -e^{i\phi}(e^{i\phi}-1)\sin\beta \cos\beta & -e^{i\phi}+(e^{i\phi}-1)\sin^2\beta
\end{bmatrix}\begin{bmatrix}
    Amp_{\ket{\tau}}\\
    Amp_{\ket{a}}
\end{bmatrix}\label{eq:Iteration Matrix}
\end{equation}

\noindent where $\phi$ is the phase rotation, and $Amp_{\ket{\tau}}$ and $Amp_{\ket{a}}$ are the amplitudes for the target element $\ket{\tau}$ and the other element $\ket{a}$ made of the non-target elements properly normalized. The angle $\beta$ that is used above and which we use for simplification later in the work is defined as

\begin{equation}
    \sin\beta = \frac{1}{\sqrt{N}},\quad\quad \cos\beta = \sqrt{\frac{N-1}{N}}.\label{eq:BetaDesc}
\end{equation}

From \eqref{eq:Iteration Matrix}, note that if the initial amplitude is known, then the amplitude and probability of observation can be calculated at any step of the iteration mathematically, which allows further optimization through classical computing.

\section{Optimizing the probability of a generic step}\label{sec4}

Before fully investigating optimal phase change for generalized Grover's algorithm, let us go over optimizing phase change for the very first step when starting from the Hadamard gate. In this simple computation, we can easily identify the classical Grover's phase change arising as optimal in the first iteration. This straightforward derivation serves as motivation and a preindication of the broader task addressed in this section.

Recall the generalized Grover's updating linear map given in \eqref{eq:Iteration Matrix} as well as the corresponding initial vector $\mathbf{v}_0^H$ given as $\mathbf{v}_0^H = [\sin\beta\ \cos\beta]^\top = [\frac{1}{\sqrt{N}}\ \sqrt{\frac{N-1}{N}}]^\top$ (see \eqref{eq:BetaDesc}) from the Hadamard gate.

From \eqref{eq:Iteration Matrix} and \eqref{eq:BetaDesc}, we acquire the probability of the target $\ket{\tau}$ as

\[P(\ket{\tau}) = |e^{i\phi}((1-e^{i\phi})\sin^2\beta -1)\sin\beta - (e^{i\phi}-1)\cos^2\beta\sin\beta|^2\]

\noindent which simplifies to

\[P(\ket{\tau})=\sin^2\beta|((1-e^{i\phi})\sin^2\beta-1)+(e^{-i\phi}-1)\cos^2\beta|^2.\]

\noindent Letting $\eta$ denote the quantity inside the norm, we compute its real and imaginary parts.

\[Re(\eta) = \sin^2\beta(1-\cos\phi)-1+\cos^2\beta(\cos\phi-1).\]
\[=\cos(2\beta)\cos\phi-\cos(2\beta)-1.\]
\[Im(\eta) = -(\sin^2\beta+\cos^2\beta)\sin\phi = -\sin\phi.\]

\noindent Thus, the target probability becomes

\[P(\ket{\tau}) = \sin^2\beta((\cos(2\beta)\cos\phi-\cos(2\beta)-1)^2+\sin^2\phi).\]

\noindent Recalling that $\sin\beta = \frac{1}{\sqrt{N}}$ and $\cos\beta=\sqrt{\frac{N-1}{N}}$, we let $a=\cos(2\beta)=1-\frac{2}{N}\in(0,1)$ and rewrite \noindent $P(\ket{\tau})$ accordingly.

\[P(\ket{\tau}) = \frac{1-a}{2}((a\cos\phi-a-1)^2+1-\cos^2\phi).\]

\noindent Since the dependence on $\phi$ is only through $\cos\phi$, by letting $u=\cos\phi\in[-1,1]$, we get the probability $P(\ket{\tau})$ as a function of $u$ on $[-1,1]$.

\[P(\ket{\tau}) = \frac{1-a}{2}((a^2-1)u^2-2a(a+1)u+(a+1)^2+1)\]
\[=-\frac{(a-1)^2(a+1)}{2}\biggl(u^2-2\frac{a}{a-1}u+\frac{(a+1)^2+1}{a^2-1}\biggr).\]

This quadratic polynomial in $u$ is maximized when $u=\frac{a}{a-1} = 1-\frac{N}{2}$. Since maximization is done on the interval $[-1,1]$, if $N\geq 4$, which ensures $\frac{a}{a-1} = 1-\frac{N}{2}\leq -1$, then the maximum is achieved when $u=-1$, in other words, $\phi=\pi$. This proves that in the generalized Grover's algorithm initialized with the Hadamard gate, the optimal phase change in the very first iteration is $\pi$, which is the phase used in the classical Grover's algorithm.

We now turn to the general case where the amplitude vector describes a generic state. To calculate the optimal angle for the probability of this generic state, we first need to derive a formula for the probability. Let us start with describing the amplitude of a generic state, and for this, we will be using the setup made in section \ref{sec2}. We remind the reader that the amplitudes satisfy the following:

\begin{enumerate}
    \item The amplitudes are complex,
    \item The sum of their modulus squared is equal to 1, i.e. $|Amp_{\ket{\tau}}|^2 + |Amp_{\ket{a}}|^2=1$.
\end{enumerate}

Therefore, we know there exists an $\alpha$ and $\theta_1,\theta_2$ in $\mathbb{R}$ such that $Amp_{\ket{\tau}} = \sin(\alpha)e^{i\theta_1}$ and $Amp_{\ket{a}} = \cos(\alpha)e^{i\theta_2}$. Using this as our amplitude vector, we get

\[\begin{bmatrix}
    \sin\alpha\ e^{i\theta_1}\\
    \cos\alpha\ e^{i\theta_2}
\end{bmatrix} = e^{i\theta_2}\begin{bmatrix}
    \sin\alpha\ e^{i(\theta_1-\theta_2)}\\
    \cos\alpha
\end{bmatrix}.\]

Then because the shared phase changes on both vectors do not impact the probability, we will drop $e^{i\theta_2}$, leaving us with our generic amplitude vector of the following form,

\begin{equation}
    \mathbf{v}_m = \begin{bmatrix}
    \sin(\alpha)e^{i\theta}\\
    \cos(\alpha)
\end{bmatrix}\label{eq:GenAmp}\end{equation}

\noindent where $\theta = \theta_1 - \theta_2$ and $\alpha\in[0,\frac{\pi}{2}]$. From \eqref{eq:Iteration Matrix} note that the generalized Grover's iteration matrix is as follows:

\[\mathbf{A}^\phi = \begin{bmatrix}
    -e^{i\phi}(1+(e^{i\phi}-1)\sin^2\beta) & -(e^{i\phi}-1)\sin\beta \cos\beta\\
    -e^{i\phi}(e^{i\phi}-1)\sin\beta \cos\beta & -e^{i\phi}+(e^{i\phi}-1)\sin^2\beta
\end{bmatrix}\]
\begin{equation}
    = \begin{bmatrix}
    e^{i\phi}\Big(\frac{1-e^{i\phi}}{N}-1\Big) & \frac{\sqrt{N-1}}{N}(1-e^{i\phi})\\
    \frac{\sqrt{N-1}}{N}e^{i\phi}(1-e^{i\phi}) & -\Big(\frac{1}{N}+\Big(1-\frac{1}{N}\Big)e^{i\phi}\Big)
\end{bmatrix}\label{eq:IterationMatrix2}.\end{equation}

\begin{theorem}
    \label{thm:GeneralProbability}
    Suppose we carry out the generalized Grover's iteration defined in \eqref{eq:IterationMatrix2} with the general initial vector given by $\mathbf{v}_m^\top = [\sin(\alpha)e^{i\theta},\ \cos(\alpha)]^\top$ from \eqref{eq:GenAmp}, then $P(\ket{\tau})$, the probability of the target vector $\ket{\tau}$, is given by the following function of $\phi$, $\theta$, $\alpha$, and $N$,
    \begin{align*}
    P(\ket{\tau}) = &\ \frac{2\sqrt{N-1}}{N}\biggl[\biggl(\frac{\sqrt{N-1}}{N}\cos(2\alpha)+\frac{1}{N}\sin(2\alpha)\cos(\phi+\theta)\biggr)(1-\cos\phi)-\\
    &-\frac{1}{2}\sin(2\alpha)\cos(\phi+\theta)\biggr]+\frac{\sqrt{N-1}}{N}\sin(2\alpha)\cos(\theta)+\sin^2\alpha.
\end{align*}
\end{theorem}

\begin{proof}

To calculate the updated probability for our target element after one iteration, we find the first element of $\mathbf{A}^\phi \mathbf{v}_m$ and square the modulus, which is

\[P({\ket{\tau}}) = \biggl|\sin\alpha\ e^{i\theta}e^{i\phi}\biggl(\frac{1-e^{i\phi}}{N}-1\biggr)+\cos\alpha\frac{\sqrt{N-1}}{N}(1-e^{i\phi})\biggr|^2\]
\smallskip
\[=\biggl|\sin\alpha\ e^{i\theta}(\sin^2\beta(1-e^{i\phi})-1)+\cos\alpha\cos\beta\sin\beta(e^{-i\phi}-1)\biggr|^2.\]

With further algebraic simplification, this reduces to

\begin{equation}\begin{split}
   P({\ket{\tau}})  =&\ \frac{1}{4}\biggl|-\sin\alpha(\cos(2\beta)+1)e^{i\theta}-\cos\alpha\sin(2\beta) + e^{i(\theta+\phi)}(\cos(2\beta)-1)\sin\alpha+ \\
    & + e^{-i\phi}\cos\alpha\sin(2\beta)\biggr|^2.\label{eq:ProbTarget}
\end{split}
\end{equation}

To organize the modulus calculation, we denote

\[\begin{split}
    \dagger =& -\sin\alpha(\cos(2\beta)+1)e^{i\theta}-\cos\alpha\sin(2\beta) + e^{i(\theta+\phi)}(\cos(2\beta)-1)\sin\alpha+\\
    &+ e^{-i\phi}\cos\alpha\sin(2\beta).
\end{split}
\]

\[Re(\dagger) = -\Big[ \sin\alpha\Big\{(\cos\theta + \cos(\phi+\theta))+\cos(2\beta)(\cos\theta-\cos(\phi+\theta))\Big\} + \cos\alpha\Big\{\sin(2\beta)(1-\cos\phi)\Big\}\Big].\]
\[Im(\dagger) = -\Big[\sin\alpha\Big\{(\sin\theta+\sin(\phi+\theta))+\cos(2\beta)(\sin\theta-\sin(\phi+\theta))\Big\}+\cos\alpha\Big\{\sin(2\beta)\sin\phi\Big\}\Big].\]

Calculating the modulus squared of $\dagger$, we get
\begin{equation*}
    \begin{split}
        |\dagger|^2 =&\ \sin^2\alpha\biggl[(2\cos(\theta+\frac{\phi}{2})\cos(\frac{\phi}{2})+2\cos(2\beta)\sin(\theta+\frac{\phi}{2})\sin(\frac{\phi}{2})\biggr]^2+\cos^2\alpha\sin^2(2\beta)(1-\cos\phi)^2 +\\
        &+2\sin\alpha\cos\alpha\biggl[(2\cos(\theta+\frac{\phi}{2})\cos(\frac{\phi}{2})+2\cos(2\beta)\sin(\theta+\frac{\phi}{2})\sin(\frac{\phi}{2})\biggr]\sin(2\beta)(1-\cos\phi) +\\
        &+\sin^2\alpha\biggl[(2\sin(\theta+\frac{\phi}{2})\cos(\frac{\phi}{2})-2\cos(2\beta)\cos(\theta+\frac{\phi}{2})\sin(\frac{\phi}{2})\biggr]^2 + \cos^2\alpha\sin^2(2\beta)\sin^2\phi +\\
        &+2\sin\alpha\cos\alpha\biggl[(2\sin(\theta+\frac{\phi}{2})\cos(\frac{\phi}{2})-2\cos(2\beta)\cos(\theta+\frac{\phi}{2})\sin(\frac{\phi}{2})\biggr]\sin(2\beta)\sin\phi.
    \end{split}
\end{equation*}

\noindent Grouping related terms together, we get

\begin{equation*}
    \begin{split}
        |\dagger|^2 =&\ 4\sin^2\alpha\biggl[\biggr((\cos(\theta+\frac{\phi}{2})\cos(\frac{\phi}{2})+\cos(2\beta)\sin(\theta+\frac{\phi}{2})\sin(\frac{\phi}{2})\biggr)^2+\\
        &+\biggl((\sin(\theta+\frac{\phi}{2})\cos(\frac{\phi}{2})-\cos(2\beta)\cos(\theta+\frac{\phi}{2})\sin(\frac{\phi}{2})\biggr)^2\biggr]+\\
        &+\cos^2\alpha\biggl[\sin^2(2\beta)(2\sin^2\frac{\phi}{2})^2+\sin^2(2\beta)(2\sin(\frac{\phi}{2})\cos(\frac{\phi}{2}))^2\biggr]+\\
        &+4\sin\alpha\cos\alpha\biggl\{\biggl[(\cos(\theta+\frac{\phi}{2})\cos(\frac{\phi}{2})+\cos(2\beta)\sin(\theta+\frac{\phi}{2})\sin(\frac{\phi}{2})\biggr]\sin(2\beta)(1-\cos\phi)+\\
        &+\biggl[(\sin(\theta+\frac{\phi}{2})\cos(\frac{\phi}{2})-\cos(2\beta)\cos(\theta+\frac{\phi}{2})\sin(\frac{\phi}{2})\biggr]\sin(2\beta)\sin\phi\biggr\}.
    \end{split}
\end{equation*}

\noindent If we use some trigonometric identities to express this in terms of $\frac{\phi}{2}$,

\begin{equation*}
    \begin{split}
        |\dagger|^2=&\ 4\sin^2\alpha\biggl[\biggl((\cos(\theta+\frac{\phi}{2})\cos(\frac{\phi}{2})+\cos(2\beta)\sin(\theta+\frac{\phi}{2})\sin(\frac{\phi}{2})\biggr)^2+\\
        &+\biggl((\sin(\theta+\frac{\phi}{2})\cos(\frac{\phi}{2})-\cos(2\beta)\cos(\theta+\frac{\phi}{2})\sin(\frac{\phi}{2})\biggr)^2\biggr]+\\
        &+4\cos^2\alpha\biggl[\sin^2(2\beta)\sin^4\frac{\phi}{2}+\sin^2(2\beta)\sin^2\frac{\phi}{2}\cos^2\frac{\phi}{2}\biggr]+\\
        &+4\sin\alpha\cos\alpha\sin(2\beta)\biggl\{\biggl[(\cos(\theta+\frac{\phi}{2})\cos(\frac{\phi}{2})+\cos(2\beta)\sin(\theta+\frac{\phi}{2})\sin(\frac{\phi}{2})\biggr](2\sin^2\frac{\phi}{2})+\\
        &+\biggl[(\sin(\theta+\frac{\phi}{2})\cos(\frac{\phi}{2})-\cos(2\beta)\cos(\theta+\frac{\phi}{2})\sin(\frac{\phi}{2})\biggr]2\sin(\frac{\phi}{2})\cos(\frac{\phi}{2})\biggr\}.
    \end{split}
\end{equation*}

Now we return to calculating the probability of our target $\ket{\tau}$.

\[P(\ket{\tau}) = \frac{1}{4}|\dagger|^2 = (\sin^2\alpha)Q_1 + (\cos^2\alpha)Q_2 + 2\sin\alpha\cos\alpha\Big(\sin(2\beta)\sin(\frac{\phi}{2})\Big)Q_3,\]

\noindent where the $Q_1$, $Q_2$, and $Q_3$ are defined as follows.
\[\begin{split}
       Q_1 =&\ \biggl\{\cos(\theta+\frac{\phi}{2})\cos(\frac{\phi}{2})+\cos(2\beta)\sin(\theta+\frac{\phi}{2})\sin(\frac{\phi}{2})\biggr\}^2+ \\
       &+\biggl\{\sin(\theta+\frac{\phi}{2})\cos(\frac{\phi}{2})-\cos(2\beta)\cos(\theta+\frac{\phi}{2})\sin(\frac{\phi}{2})\biggr\}^2,
\end{split}
\]

\[Q_2=\sin^2(2\beta)\sin^4(\frac{\phi}{2})+\sin^2(2\beta)\sin^2(\frac{\phi}{2})\cos^2(\frac{\phi}{2}),\]

\[\begin{split}
    Q_3  =&\ \sin(\frac{\phi}{2})\biggl(\cos(\theta+\frac{\phi}{2})\cos(\frac{\phi}{2})+\cos(2\beta)\sin(\theta+\frac{\phi}{2})\sin(\frac{\phi}{2})\biggr)+ \\
    & +\cos(\frac{\phi}{2})\biggl(\sin(\theta+\frac{\phi}{2})\cos(\frac{\phi}{2})-\cos(2\beta)\cos(\theta+\frac{\phi}{2})\sin(\frac{\phi}{2})\biggr).
\end{split}\]

\noindent Below we simplify these three terms ($Q_1$, $Q_2$, and $Q_3$) a little more.

\begin{equation*}
    \begin{split}
        Q_1 =&\ \cos^2(\theta+\frac{\phi}{2})\cos^2\frac{\phi}{2} +\cos^2(2\beta)\sin^2(\theta+\frac{\phi}{2})\sin^2\frac{\phi}{2}+\\
        &+\sin^2(\theta+\frac{\phi}{2})\cos^2\frac{\phi}{2}+\cos^2(2\beta)\cos^2(\theta+\frac{\phi}{2})\sin^2\frac{\phi}{2}+\\
        &+ 2\cos(\theta+\frac{\phi}{2})\cos(\frac{\phi}{2})\cos(2\beta)\sin(\theta+\frac{\phi}{2})\sin(\frac{\phi}{2})-\\
        &- 2\sin(\theta+\frac{\phi}{2})\cos(\frac{\phi}{2})\cos(2\beta)\cos(\theta+\frac{\phi}{2})\sin(\frac{\phi}{2})
    \end{split}
\end{equation*}
\[=\cos^2\frac{\phi}{2}+\cos^2(2\beta)\sin^2\frac{\phi}{2}\]
\[=\cos^2\frac{\phi}{2}+(1-\sin^2(2\beta))\sin^2\frac{\phi}{2}\]
\[=1-\sin^2(2\beta)\sin^2(\frac{\phi}{2}).\]
\[Q_2 = \sin^2(2\beta)\sin^2(\frac{\phi}{2}).\]
\[Q_3 = \cos(\theta+\frac{\phi}{2})\sin(\frac{\phi}{2})\cos(\frac{\phi}{2})(1-\cos(2\beta))+\sin(\theta+\frac{\phi}{2})(\cos(2\beta)\sin^2(\frac{\phi}{2})+\cos^2(\frac{\phi}{2}))\]
\[=\cos(\theta+\frac{\phi}{2})\sin(\frac{\phi}{2})\cos(\frac{\phi}{2})(1-\cos(2\beta))+\sin(\theta+\frac{\phi}{2})(\sin^2(\frac{\phi}{2})\cos(2\beta)+(1-\sin^2(\frac{\phi}{2}))\]
\[=\cos(\theta+\frac{\phi}{2})\sin(\frac{\phi}{2})\cos(\frac{\phi}{2})(1-\cos(2\beta))+\sin(\theta+\frac{\phi}{2})(-\sin^2(\frac{\phi}{2})(1-\cos(2\beta))+1)\]
\[=(1-\cos(2\beta))\sin(\frac{\phi}{2})\biggl[\cos(\theta+\frac{\phi}{2})\cos(\frac{\phi}{2})-\sin(\theta+\frac{\phi}{2})\sin(\frac{\phi}{2})\biggr]+\sin(\theta+\frac{\phi}{2})\]
\[=\cos(\theta+\phi)\sin(\frac{\phi}{2})(1-\cos(2\beta))+\sin(\theta+\frac{\phi}{2})
.\]
Putting these all together, our target probability becomes

\[\begin{split}
    P(\ket{\tau})  =&\ \sin^2\alpha\biggl[1-\sin^2(2\beta)\sin^2(\frac{\phi}{2})\biggr]+\cos^2\alpha\sin^2(2\beta)\sin^2(\frac{\phi}{2})+ \\
    & +2\sin\alpha\cos\alpha\sin(2\beta)\sin(\frac{\phi}{2})\biggl[\cos(\theta+\phi)\sin(\frac{\phi}{2})(1-\cos(2\beta))+\sin(\theta+\frac{\phi}{2})\biggr]
    \end{split}\]

\[\begin{split}
    =&\ \sin^2\alpha + \sin^2(2\beta)\cos(2\alpha)\sin^2(\frac{\phi}{2})+\sin(2\alpha)\sin(2\beta)\biggl[2\cos(\phi+\theta)\sin^2(\frac{\phi}{2})\sin^2\beta +\\
    &+ \sin(\frac{\phi}{2}+\theta)\sin(\frac{\phi}{2})\biggr],
\end{split}\]

\noindent which reduces to

\begin{equation}\begin{split}
P(\ket{\tau})= &\ \sin^2\alpha+\sin^2(2\beta)\cos(2\alpha)\biggl(\frac{1-\cos\phi}{2}\biggr)+ \\
&+\sin(2\alpha)\sin(2\beta)\biggl[\cos(\phi+\theta)(1-\cos\phi)\sin^2\beta-\frac{1}{2}\cos(\phi+\theta)+\frac{1}{2}\cos\theta\biggr].\label{eq:ProbTargetUpdate}
\end{split}\end{equation}

We then use \eqref{eq:BetaDesc} to incorporate $\sin\beta = \frac{1}{\sqrt{N}}$ and $\sin(2\beta) = \frac{2\sqrt{N-1}}{N}$. Substituting these values into \eqref{eq:ProbTargetUpdate},

\begin{equation}
    \begin{split}
        P(\ket{\tau}) = &\ \frac{2\sqrt{N-1}}{N}\biggl[\biggl(\frac{\sqrt{N-1}}{N}\cos(2\alpha)+\sin(2\alpha)\frac{1}{N}\cos(\phi+\theta)\biggr)(1-\cos\phi)-\\
    &-\frac{1}{2}\sin(2\alpha)\cos(\phi+\theta)\biggr]+\frac{\sqrt{N-1}}{N}\sin(2\alpha)\cos(\theta)+\sin^2\alpha.
    \end{split}\label{eq:OneStepProb}
\end{equation}

This completes the proof.
\end{proof}

Now we can address our goal to find the optimal phase to maximize $P(\ket{\tau})$, the probability of the target vector $\ket{\tau}$ in the following Corollary.

\begin{corollary}
\label{cor:argmax}
    For the one-step iteration of the generalized Grover's algorithm in equation \eqref{eq:IterationMatrix2} with the general initial vector $\mathbf{v}_m^\top = [\sin(\alpha)e^{i\theta},\ \cos(\alpha)]^\top$, the optimal phase change for the target probability $P(\ket{\tau})$ is given as follows
    \begin{equation}\begin{split}
    \argmax_{\phi\in[-\pi,\pi]}  \biggl(&\frac{\sqrt{N-1}}{N}\cos(2\alpha)+\frac{1}{N}\sin(2\alpha)\cos(\phi+\theta)\biggr)(1-\cos(\phi))\\ & -\frac{1}{2}\sin(2\alpha)\cos(\phi+\theta).\label{eq:Optimizaiton}
\end{split}
\end{equation}
    This optimal phase change is identified in Figure \ref{fig:1}.
\end{corollary}

\begin{proof}
    The statement of \eqref{eq:Optimizaiton} can be seen directly by examination of equation \eqref{eq:OneStepProb}. Recall that $\alpha$ and $\theta$ are from the amplitude vector $\mathbf{v}_m$ before update.

    To visualize this computationally, we utilize the optimization and global optimization toolboxes of MATLAB. Because $\alpha$ and $\theta$ are values that come from the description of the original amplitude vector, we create a grid of 100 values of $\alpha$ and of $\theta$ ranging over $[0,\pi/2]$ and $[-\pi,\pi]$ respectively. 
    
We use a Golden-section search (fminbnd) due to its speed when doing many examples. To account for the fact that this might give us local maximums, we utilized the global optimization toolbox to run interior-point methods for randomly selected starting points to find our global maximum, or at least provide a better optimization than the golden-section search. What we found is that the differences between these two searches were negligible as seen in Figure \ref{fig:1}.

\begin{figure}[h!]
    \centering
    \includegraphics[width=1\linewidth]{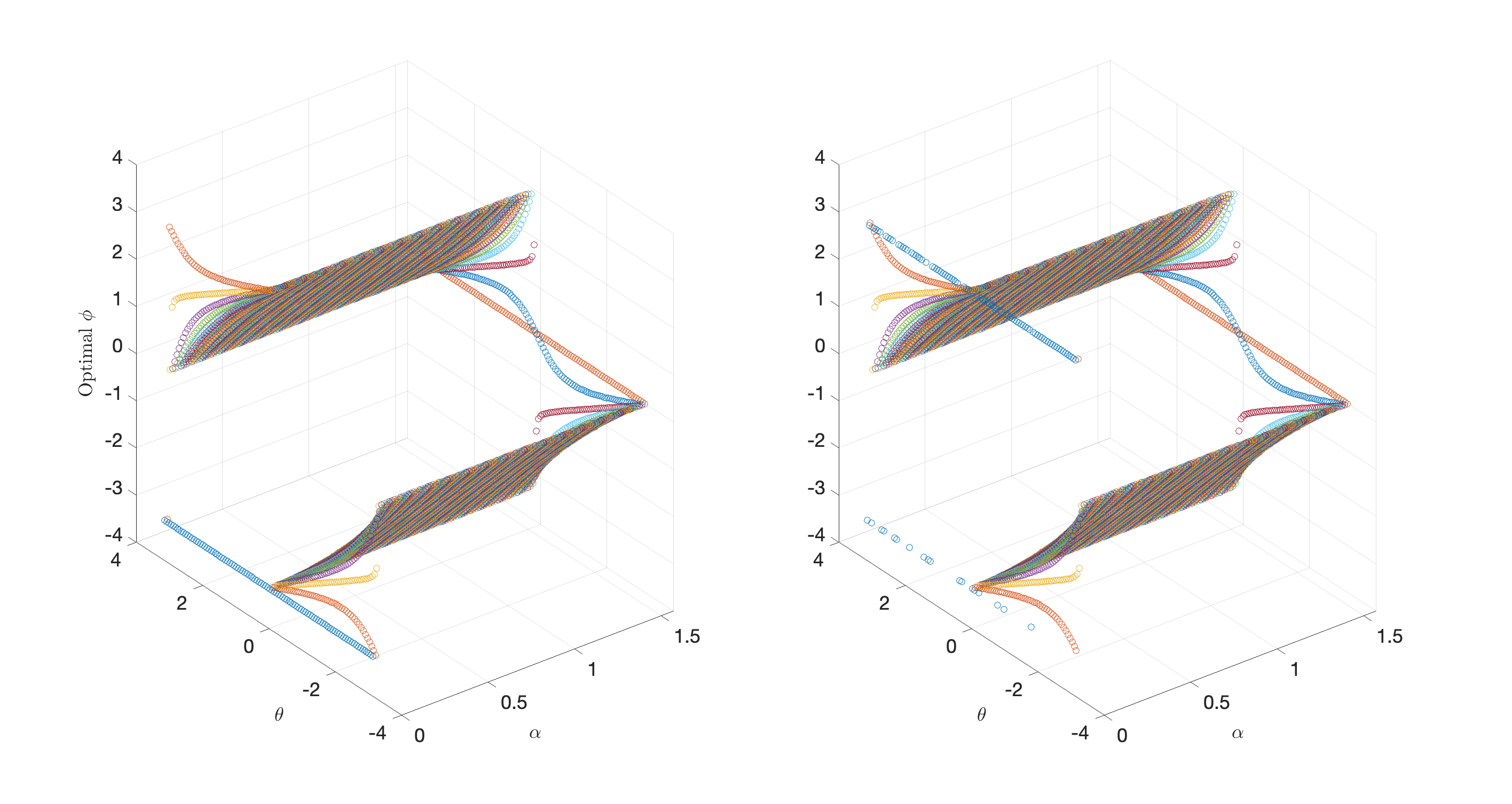}
    \caption{Comparison of optimal phase change found by local search (left) and global search (right) for $N = 2^{10}$. The coloring of each point is used for visual distinction and does not have any further meaning.}
    \label{fig:1}
\end{figure}
\end{proof}

\begin{remark}
    As can be seen in Figure \ref{fig:1}, the optimal $\phi$ is in fact related to the value of $\theta$ found in \eqref{eq:GenAmp}, the amplitude vector. This relationship roughly is

\begin{equation}
    \begin{cases}
    \phi \approx -\pi-\theta,\quad \theta < 0 \\
    \phi \approx \pi-\theta,\quad \ \ \theta > 0
\end{cases}\label{eq:RoughRelationship}
\end{equation}
\end{remark}

This generalization does not hold when $\alpha$ is an extreme value, either being too close to 0 (corresponding to an amplitude close to 0 for our target), or too close to $\pi/2$ (corresponding to an amplitude close to 1 for our target). In these cases, the optimal $\phi$ cannot be estimated in the same way, and does require use of an optimization algorithm to find.

It is important to note that the relationship described in \eqref{eq:RoughRelationship} is a rough estimate. While it is a good approximation, the global search does differ, and the difference is shown in Figure \ref{fig:rough_change} where we compare the optimal phase identified in Corollary \ref{cor:argmax} with

\begin{equation}
    \phi = \begin{cases}
    -\pi-\theta,\quad \theta <0\\
    \pi-\theta,\quad\ \ \theta >0
\end{cases}\label{eq:SimplePhi}
\end{equation}

\begin{figure}[h!]
    \centering
    \includegraphics[width=1\linewidth]{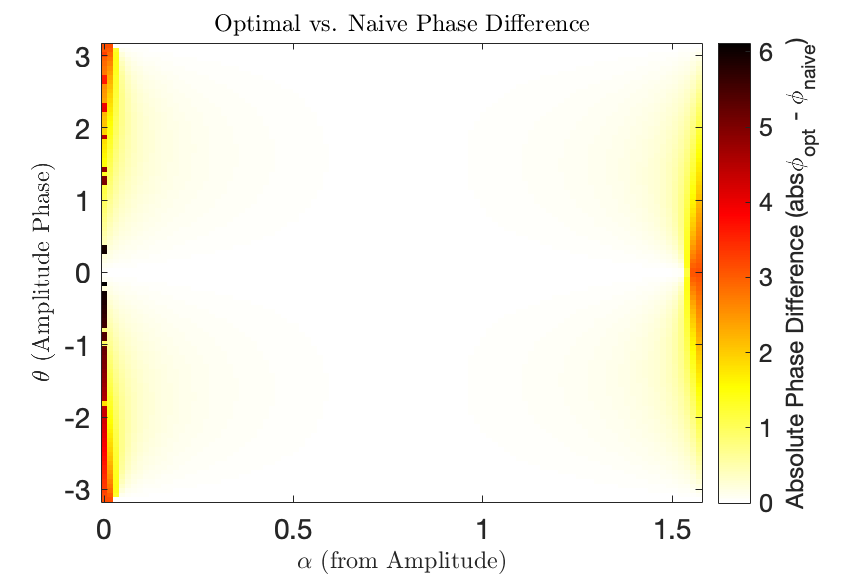}
    \caption{A heatmap of the absolute difference between the optimal phase change found with the global search of \eqref{eq:Optimizaiton} and the naive phase change determined by \eqref{eq:SimplePhi} for a system of size $N=2^{10}$.}
    \label{fig:rough_change}
\end{figure}

\begin{remark}
    By using the optimal $\phi$ as the phase change rather than $\pi$, we find a range of improvement in the probability increase provided by one iteration. This improvement depends on the amplitude vector's $\alpha$ and $\theta$ values. This comparison results in Figure \ref{fig:2}.
\end{remark}

\begin{figure}[h!]
    \centering
    \includegraphics[width=1\linewidth]{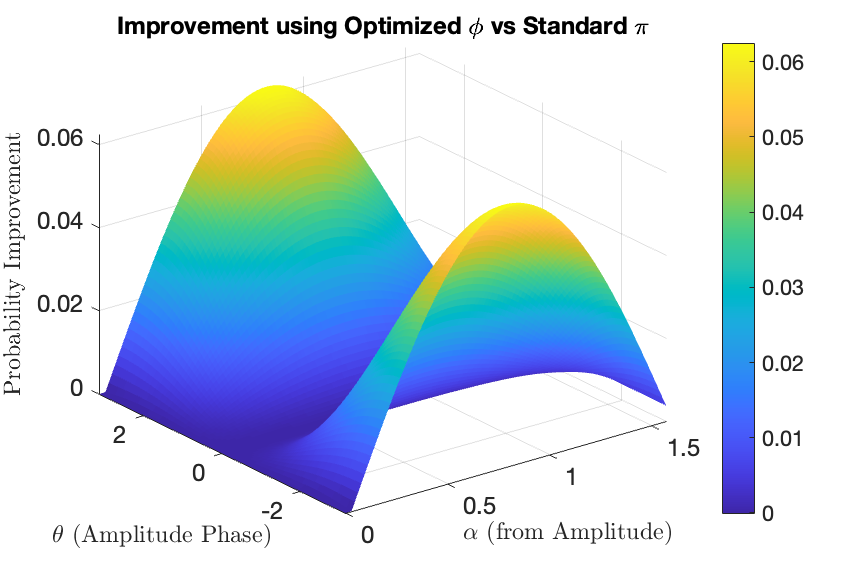}
    \caption{A graph of the difference in $P(\ket{\tau})$ after one iteration using the optimal $\phi$ found with the global search of \eqref{eq:Optimizaiton} and of using traditional Grover's algorithm with a phase change of $\pi$ for a system of size $N=2^{10}$.}
    \label{fig:2}
\end{figure}

As can be seen in Figure \ref{fig:2}, when $\theta$ is close to 0, there is little improvement over the probability attained by Grover's algorithm. However, as $\theta$ moves away from 0, the improvement provided by the optimized phase change increases, with the bigger improvements as $\theta$ gets closer to $\pi$. This agrees with the relationship seen in Figure \ref{fig:1} and also observed in \eqref{eq:RoughRelationship}, since they suggest an optimal phase change would be close to $\phi = 0$ when $\theta$ gets close to $\pi$, while Grover's algorithm would be using a phase change of $\phi = \pi$.

In the next section, we use this framework to understand the optimality of classical Grover's algorithm in much more precise terms, and determine when the optimal phase change strays away from $\pi$, the phase from classical Grover's algorithm.

\section{Optimizing when initializing with the Hadamard gate}\label{sec5}

Adopting the Hadamard gate at the beginning is part of the standard procedure when using Grover's algorithm since it initializes all states with the same amplitude, mimicking the lack of underlying structure in the system (or potentially, our lack of understanding information about it). This leaves us with an initial amplitude vector based on $\ket{\tau}$ and $\ket{a}$

\begin{equation}
    \mathbf{v}_0 = \begin{bmatrix}
    \frac{1}{\sqrt{N}}\\
    \\
    \sqrt{\frac{N-1}{N}}
\end{bmatrix}=\mathbf{v}_0^H.\label{eq:InitialHadamard}
\end{equation}

As can be seen in \eqref{eq:InitialHadamard}, the vector is completely real, i.e. $\mathbf{v}_0\in\mathbb{R}^2$, hence our $\theta = 0$, when compared to the general amplitude vector \eqref{eq:GenAmp}. Recall that $|v_1|^2 + |v_2|^2 = 1$ for $\mathbf{v} = [v_1\ v_2]^\top$. Let us first consider a general but real amplitude vector. First of all, note the Hadamard gate, which is how most of the unstructured search algorithms begin in general, is real, and we want to first explore how and when the vector is no longer real, while updated with the optimal phase change.

\begin{theorem}
\label{th:RealVec}
    Consider a general but real initial vector $\textbf{v}=[\sin(\alpha),\ \cos(\alpha)]^\top$ to start the generalized Grover's iteration given in \eqref{eq:IterationMatrix2}. Then the optimal phase change for the target probability $P(\ket{\tau})$ behaves differently depending on the region $\alpha$ belongs to of three different regions $\{\textbf{R}_i\}_{1\leq i\leq 3}$. More precisely,

    \[\argmax_{\phi\in [-\pi,\pi]} P(\ket{\tau}) = \begin{cases}
        \pi,\ \alpha \in \textbf{R}_1\\
        \cos^{-1}(\frac{1}{4}(-N+2-2\sqrt{N-1}\cot(2\alpha))),\ \alpha\in \textbf{R}_2\\
        0,\ \alpha\in\textbf{R}_3
    \end{cases}\]

    where

    \[\begin{cases}
        \textbf{R}_1 = \Bigl(0,\frac{1}{2}\cot^{-1}\Big(\frac{-N+6}{2\sqrt{N-1}}\Big)\Bigr]\\
    \textbf{R}_2 = \Big(\frac{1}{2}\cot^{-1}\Big(\frac{-N+6}{2\sqrt{N-1}}\Big),\frac{1}{2}\cot^{-1}\Big(\frac{-N-2}{2\sqrt{N-1}}\Big)\Big)\\
    \textbf{R}_3 = \Bigl[\frac{1}{2}\cot^{-1}\Big(\frac{-N-2}{2\sqrt{N-1}}\Big), \frac{\pi}{2}\Bigr)
    \end{cases}\]
\end{theorem}

\begin{proof}
Consider a general but real initial vector

    \[\mathbf{v} = \begin{bmatrix}
    \sin\alpha\\
    \cos\alpha
\end{bmatrix}.\]

We then apply Grover's iterate matrix \eqref{eq:IterationMatrix2} to this general but real amplitude vector. Then the updated amplitude vector is

\begin{equation}
    \begin{bmatrix}
    e^{i\phi}\Big(\frac{1-e^{i\phi}}{N}-1\Big)\sin\alpha + \frac{\sqrt{N-1}}{N}(1-e^{i\phi})\cos\alpha\\
    \\
    \frac{\sqrt{N-1}}{N}e^{i\phi}(1-e^{i\phi})\sin\alpha - \frac{1}{N}(1+(N-1)e^{i\phi})\cos\alpha
\end{bmatrix}.\label{eq:OneStepAmp}
\end{equation}

Now we look at the real and imaginary parts of each coordinate in \eqref{eq:OneStepAmp}. Note that the imaginary parts for both are multiplied by either $\sin\phi$ or $\sin(2\phi)$. Therefore, if $\phi = \pi$, as in the classical Grover's algorithm, the imaginary part becomes 0. This means that if the amplitude vector is real, and a Grover iterate matrix with a phase change of $\pi$ (i.e. classical Grover's iterate) is applied to it, the resulting amplitude vector will remain real. Hence if the system is initialized with a Hadamard gate, giving us the real initial amplitude vector \eqref{eq:InitialHadamard}, the amplitude vector will remain real throughout all of the iterations.

With this in mind, we return to our amplitude vector after one iteration \eqref{eq:OneStepAmp}. We can calculate the resulting probability of observing our target

\[P(\ket{\tau}) = \biggl|e^{i\phi}\biggl(\frac{1-e^{i\phi}}{N}-1\biggr)\sin\alpha + \frac{\sqrt{N-1}}{N}(1-e^{i\phi})\cos\alpha\biggr|^2.\]

\noindent Substituting in our notation with $\beta$ in \eqref{eq:BetaDesc} and factoring out $-e^{i\phi}$ from the amplitude, we get

\begin{equation}
    P(\ket{\tau}) = \Big|((e^{i\phi}-1)\sin^2\beta+1)\sin\alpha+\sin\beta\cos\beta(1-e^{-i\phi})\cos\alpha\Big|^2.\label{eq:ProbHGate}
\end{equation}

\noindent To calculate the probability, we break down the transformed amplitude as follows.

\[\xi = ((e^{i\phi}-1)\sin^2\beta+1)\sin\alpha+\sin\beta\cos\beta(1-e^{-i\phi})\cos\alpha.\]
\[Re(\xi) = A\cos\phi + B,\]
\[Im(\xi) = C\sin\phi,\]

\noindent where

\[A = \sin^2\beta\sin\alpha-\sin\beta\cos\beta\cos\alpha,\]
\[B = \sin\alpha\cos^2\beta+\cos\alpha\sin\beta\cos\beta,\]
\[C = \sin^2\beta\sin\alpha + \sin\beta\cos\beta\cos\alpha.\]

\noindent We can then rewrite the probability \eqref{eq:ProbHGate} as

\[P(\ket{\tau}) = (A^2-C^2)\cos^2\phi + 2AB\cos\phi + B^2 + C^2.\]

\noindent To simplify our notation, let us define $a$ and $b$ as follows.

\begin{equation*}
    \begin{split}
    a = C^2-A^2 =&\   \sin^4\beta\sin^2\alpha + \sin^2\beta\cos^2\beta\cos^2\alpha + 2\sin^3\beta\sin\alpha\cos\alpha\cos\beta-\\
    & -\sin^4\beta\sin^2\alpha-\sin^2\beta\cos^2\beta\cos^2\alpha+2\sin^3\beta\sin\alpha\cos\alpha\cos\beta
\end{split}
\end{equation*}
\[=4\sin^2\beta(\sin\beta\cos\beta)(\sin\alpha\cos\alpha)\]
\[=\sin^2\beta\sin(2\beta)\sin(2\alpha)\]
\[b = -2AB = -2\sin\beta\cos\beta(\sin\beta\sin\alpha-\cos\beta\cos\alpha)(\sin\alpha\cos\beta+\cos\alpha\sin\beta)\]
\[=\sin(2\beta)\cos(\alpha+\beta)\sin(\alpha+\beta)\]
\[=\frac{1}{2}\sin(2\beta)\sin(2\alpha+2\beta).\]

\noindent Note that $a\geq 0$ since $\alpha\in [0,\frac{\pi}{2}].$ Letting $u = \cos\phi$, $P(\ket{\tau})$ becomes a quadratic polynomial in $u$ which is concave down.

\[P(\ket{\tau}) = -a\ u^2-b\ u+B^2+C^2\]
\[= -a\biggl(u+\frac{b}{2a}\biggr)^2+\textit{constant}.\]

Note that if $|\frac{b}{2a}|\geq 1$, then the optimality for maximization is acquired at the boundary, that is when $u=-1$ or $u=1$, meaning $\phi$ is $\pi$ or 0. Hence, we now determine when we achieve the condition $|\frac{b}{2a}|<1$ since this is when the optimal $\phi$ maximizing $P(\ket{\tau})$ differs from $\pi$ or $0$. Recall that $\sin\beta = \frac{1}{\sqrt{N}}$, $\sin(2\beta) = \frac{2\sqrt{N-1}}{N}$, and see that $\cos(2\beta) = \frac{N-2}{N}$.

\[\frac{b}{2a} = \frac{\frac{1}{2}\sin(2\beta)\sin(2\alpha+2\beta)}{2\sin^2\beta\sin(2\beta)\sin(2\alpha)}.\]

\noindent By utilizing the angle addition identity, this turns into

\[\frac{b}{2a}=\frac{N}{4}\biggl(\frac{\sin(2\alpha)\cos(2\beta)+\cos(2\alpha)\sin(2\beta)}{\sin(2\alpha)}\biggr)\]
\[=\frac{N}{4}(\cos(2\beta)+\cot(2\alpha)\sin(2\beta)).\]

\noindent Again recalling how $\beta$ is defined, we can simplify $\frac{b}{2a}$ further as

\[\frac{b}{2a}=\frac{N}{4}\biggl(\frac{N-2}{N}+\frac{2\sqrt{N-1}}{N}\cot(2\alpha)\biggr)\]
\[=\frac{1}{4}(N-2+2\sqrt{N-1}\cot(2\alpha)).\]

\noindent Returning to the inequality $|\frac{b}{2a}|<1$, the equivalent inequality becomes

\[-4<N-2+2\sqrt{N-1}\cot(2\alpha)<4.\]

\noindent This is equivalent to the following condition for the angle $\alpha$.

\begin{equation}
    \frac{1}{2}\cot^{-1}\biggl(\frac{-N+6}{2\sqrt{N-1}}\biggr)<\alpha<\frac{1}{2}\cot^{-1}\biggl(\frac{-N-2}{2\sqrt{N-1}}\biggr)\label{eq:NonTrivAl}.
\end{equation}

\noindent This leads us to consider the following three regions $\{\textbf{R}_i\}_{i\leq 3}$ for $\alpha$,

\begin{equation}
    \begin{cases}
        \textbf{R}_1 = \Bigl(0,\frac{1}{2}\cot^{-1}\Big(\frac{-N+6}{2\sqrt{N-1}}\Big)\Bigr]\\
    \textbf{R}_2 = \Big(\frac{1}{2}\cot^{-1}\Big(\frac{-N+6}{2\sqrt{N-1}}\Big),\frac{1}{2}\cot^{-1}\Big(\frac{-N-2}{2\sqrt{N-1}}\Big)\Big)\\
    \textbf{R}_3 = \Bigl[\frac{1}{2}\cot^{-1}\Big(\frac{-N-2}{2\sqrt{N-1}}\Big), \frac{\pi}{2}\Bigr)
    \end{cases}\label{eq:Regions}
\end{equation}

Note that if $\alpha\in \textbf{R}_1$, then $\cot (2\alpha)\geq\frac{-N+6}{2\sqrt{N-1}}$, hence $\frac{1}{4}(N-2+2\sqrt{N-1}\cot(2\alpha))\geq1$. This means $\frac{-b}{2a}\leq-1$, hence $P(\ket{\tau})$ is maximized when $u=-1$, which corresponds to $\phi=\pi$. This implies that in the region $\textbf{R}_1$, classical Grover's algorithm is optimal. But once $\alpha\in\textbf{R}_2$, then $P(\ket{\tau})$ is maximized by a non-trivial phase change $\phi_{max}(\alpha)$, and we will identify this $\phi_{max}(\alpha)$ ($=\argmax\ P(\ket{\tau})$) more precisely below (see \eqref{eq:PhiMax}). Finally, when $\alpha\in\textbf{R}_3$, then $\cot(2\alpha)\leq\frac{-N-2}{2\sqrt{N-1}}$, hence $\frac{1}{4}(N-2+2\sqrt{N-1}\cot(2\alpha))\leq-1$. This means $\frac{-b}{2a}\geq 1$, hence $P(\ket{\tau})$ is maximized when $u=1$, which corresponds to $\phi=0$. This means no rotation should occur once $\alpha$ comes into this region.

Now let us identify the precise optimal phase change necessary in case $\alpha\in\textbf{R}_2$. Recall that if $\alpha\in\textbf{R}_2$, then $\phi_{max}(\alpha)$ is such that 

\[\cos(\phi_{max}(\alpha)) = \frac{-b}{2a} = \frac{1}{4}(-N+2-2\sqrt{N-1}\cot(2\alpha)).\] 

Since $\alpha\in[0,\frac{\pi}{2}]$, 

\begin{equation}
    \phi_{max}(\alpha) = \cos^{-1}(\frac{1}{4}(-N+2-2\sqrt{N-1}\cot(2\alpha)))\label{eq:PhiMax}.
\end{equation}
\end{proof}

\begin{corollary}
\label{cor:probtarget}

Suppose we begin the generalized Grover's iteration from the Hadamard vector $\mathbf{v}_0^H = \bigg[\frac{1}{\sqrt{N}},\ \sqrt{\frac{N-1}{N}}\bigg]^\top$, and continue to use the optimal phase change at each step to maximize the target probability $P(\ket{\tau})$. Then the classical Grover's algorithm stays optimal until the target probability $P(\ket{\tau})$ reaches the threshold $P_r(N) = \frac{1}{2}\biggl(1+\frac{N-6}{\sqrt{N^2-8N+32}}\biggr)$.
\end{corollary}
\begin{proof}
    Recall that starting from a Hadamard initial condition, of which $\theta = 0$, the optimal phase change for generalized Grover's algorithm stays $\pi$ (i.e. classical Grover's algorithm) until $\alpha$ reaches $\frac{1}{2}\cot^{-1}(\frac{-N+6}{2\sqrt{N-1}})$. In other words, the first non-trivial optimal phase differing from that of classical Grover's begins with $\alpha_0 \equiv \frac{1}{2}\cot^{-1}(\frac{-N+6}{2\sqrt{N-1}})$. Noting that the target probability is $\sin^2\alpha_0$, this shows that the optimal phase change stays with classical Grover's algorithm till the target probability reaches $P(\ket{\tau})=\Big[\sin(\frac{1}{2}\cot^{-1}(\frac{-N+6}{2\sqrt{N-1}}))\Big]^2$. We then see that
    \[P(\ket{\tau})=\sin^2\biggl(\frac{1}{2}\cot^{-1}\biggl(\frac{-N+6}{2\sqrt{N-1}}\biggr)\biggr) = \frac{1}{2}\biggl(1+\frac{N-6}{\sqrt{N^2-8N+32}}\biggr).\]
\end{proof}

\begin{remark}
    Note that the threshold $P_r(N) = \frac{1}{2}\bigl(1+\frac{N-6}{\sqrt{N^2-8N+32}}\bigr)$ converges to 1 as $N\to \infty$. One can see that $P_r(N)$ converges to 1 in the order of $1/N$. This shows that the classical Grover's algorithm remains optimal until $P(\ket{\tau})$ gets fairly close to one, more precisely until it falls within a distance of order $1/\sqrt{N}$ from one. After which it would enter $R_2$ or $R_3$ and the optimal phase change becomes non-trivial away from $\pi$.
\end{remark}

\begin{remark}
    We also refer the reader back to Figure \ref{fig:rough_change}, with particular attention to the $\theta=0$ line, which corresponds to an initialization by the Hadamard gate. This provides a qualitative visualization of the different regions from Theorem \ref{th:RealVec}, where the white area denotes $R_1$, while the colored region near $\alpha = \pi/2$ represents $R_2$ and $R_3$, for which the optimal phase change deviates from $\pi$.
\end{remark}


We now present numerical results to examine these regions further. We begin with a system of size $N=2^5$. From the region definitions in \eqref{eq:Regions}, we identify different regions for $\alpha$ values as below.

\[
\begin{cases}
    \textbf{R}_1 = [0, 1.3685]\\
    \textbf{R}_2 = [1.3685, 1.4125]\\
    \textbf{R}_3 = [1.4125, \pi/2]
\end{cases}
\]

Using these regions, we execute the generalized Grover's algorithm with the corresponding optimal phase change for the optimal number of steps $\lfloor\frac{\pi}{4}\sqrt{N}\rfloor$ followed by one additional iteration.

\begin{table}[h!]
\centering
\begin{tabular}{cccccc}\toprule
    & \multicolumn{3}{c}{Target} &  & \\ \cmidrule(lr){2-4}
     Step & $\alpha$ & Amplitude & Probability & Region (current step) & Phase Change (to reach)\\ \midrule
    0 & 0.1777 & 0.1768 & 0.0313 & $\textbf{R}_1$ & -\\
    1 & 0.5331 & 0.5082 & 0.2583 & $\textbf{R}_1$ & $\pi$\\
    2 & 0.8886 & 0.7762 & 0.6024 & $\textbf{R}_1$ & $\pi$\\
    3 & 1.2440 & 0.9471 & 0.8969 & $\textbf{R}_1$ & $\pi$\\
    4 & 1.5422 & 0.9996 & 0.9992 & $\textbf{R}_3$ & $\pi$\\
    5 & -1.5422 & -0.9996 & 0.9992 & $\textbf{R}_3$ & 0\\
\end{tabular}
\caption{Generalized Grover's algorithm run for a system of size $N=2^5$ starting from the initial state at step 0 to one after the optimal number of steps with step 5. Each step identifies the region the amplitude belongs to and the phase change chosen to reach that step.}
\label{tab:25general}
\end{table}

As shown in Table \ref{tab:25general} for a system with size $N=2^5$, the $\alpha$ value remains in region $R_1$ from the initial step through the third, the step before the optimal number of steps. Consequently by Corollary \ref{cor:probtarget}, the optimal phase change continues to be $\pi$. Classical Grover's algorithm can be applied till step 4, when the target state's $\alpha$ value enters region $R_3$. At this stage, the optimal phase change falls to zero, which results in the $\alpha$ value in step 5 to be negative of that in the previous step. This demonstrates that under the greedy optimization setup, the target state bypasses the second region $R_2$, and becomes essentially confined in the third region $R_3$, unable to increase the target probability any more.

We now repeat this same analysis but for a larger system, one of size $N=2^{10}$.

\[
\begin{cases}
    \textbf{R}_1 = [0, 1.5394]\,,\\
    \textbf{R}_2 = [1.5394, 1.5397]\,,\\
    \textbf{R}_3 = [1.5397, \pi/2]\,.
\end{cases}
\]

\begin{table}
\centering
\begin{tabular}{cccccc}\toprule
    & \multicolumn{3}{c}{Target} &  & \\ \cmidrule(lr){2-4}
     Step & $\alpha$ & Amplitude & Probability & Region (current step) & Phase Change (to reach)\\ \midrule
    0 & 0.0313 & 0.0312 & 0.000977 & $\textbf{R}_1$ & -\\
    1 & 0.0938 & 0.0936 & 0.0088 & $\textbf{R}_1$ & $\pi$\\
    2 & 0.1563 & 0.1556 & 0.0242 & $\textbf{R}_1$ & $\pi$\\
    \vdots &  & \vdots &  &  & \vdots \\
    24 & 1.5315 & 0.9992 & 0.9985 & $\textbf{R}_1$ & $\pi$\\
    25 & 1.5476 & 0.9997 & 0.9995 & $\textbf{R}_3$ & $\pi$\\
    26 & -1.5476 & -0.9997 & 0.9995 & $\textbf{R}_3$ & 0\\
\end{tabular}
\caption{Generalized Grover's algorithm run for a system of size $N=2^{10}$ starting from the initial state at step 0 to one after the optimal number of steps with step 25. Each step identifies the region the amplitude belongs to and the phase change chosen to reach that step.}
\label{tab:210general}
\end{table}

As shown by the analysis presented for Table \ref{tab:25general}, the results for the system of size $2^{10}$ in Table \ref{tab:210general} lead to the same conclusion. The target state bypasses the second region $R_2$ and ultimately becomes trapped in the third region, $R_3$.

With the results previously discussed in Tables \ref{tab:25general} and \ref{tab:210general}, we wish to better understand why the algorithm is bypassing $R_2$ and going straight to $R_3$. To do this, we will show that the size of $R_2$ is significantly smaller than the size of $R_3$ and that the step size of Grover's algorithm makes it highly likely for $R_2$ to be missed and the state to end up in $R_3$. For this we will need to work with the assumption that $N$ is large enough such that $\sin(1/\sqrt{N})\approx 1/\sqrt{N}$, however since $N=2^n$, this is an appropriate approximation for as small as $N=2^3$.

\begin{theorem}
    Let the state vector be initialized with the Hadamard gate. Then the size of $R_2$ as defined in \eqref{eq:Regions} is in the order of $O(N^{-3/2})$, and the proportion of $|R_2|$ to $|R_2\cup R_3|$ is in the order of $O(1/N)$.
\end{theorem}
\begin{proof}
    Let the state vector be $\textbf{v}_0^H = [1/\sqrt{N}\quad \sqrt{N-1}/\sqrt{N}]^\top$ from \eqref{eq:InitialHadamard}. Recall the definition for $R_2$ from \eqref{eq:Regions} which is
    
    $$\biggl(\frac{1}{2}\cot^{-1}\biggl(\frac{-N+6}{2\sqrt{N-1}}\biggr),\frac{1}{2}\cot^{-1}\biggl(\frac{-N-2}{2\sqrt{N-1}}\biggr)\biggr).$$

    We can then calculate the size of $R_2$.

    $$|R_2| = \frac{1}{2}\cot^{-1}\biggl(\frac{-N-2}{2\sqrt{N-1}}\biggr)-\frac{1}{2}\cot^{-1}\biggl(\frac{-N+6}{2\sqrt{N-1}}\biggr)$$
    $$=\frac{1}{2}\biggl(\tan^{-1}\biggl(\frac{2\sqrt{N-1}}{-N-2}\biggr)-\tan^{-1}\biggl(\frac{2\sqrt{N-1}}{-N+6}\biggr)\biggr)$$.

    By mean value theorem, we know that there exists a $\hat{\theta}_N$ in $(\frac{2\sqrt{N-1}}{-N+6},\frac{2\sqrt{N-1}}{-N-2})$ such that

    $$|R_2| = \frac{2\sqrt{N-1}}{2}\frac{1}{1+(\hat{\theta}_N)^2}\biggl(\frac{1}{-N-2}-\frac{1}{-N+6}\biggr)$$
    $$\leq \frac{8\sqrt{N-1}}{N^2-4N-12},$$

    \noindent since $1+(\hat{\theta}_N)^2 > 1$. Hence we have $|R_2| \leq \frac{8\sqrt{N-1}}{N^2-4N-12}$ which is of order $O(N^{-3/2})$.

    Next, let us find the size of $R_2\cup R_3$. To do this, we will find the distance between $\pi/2$ and the right hand bound of $R_1$.

    $$|R_2\cup R_3| = \frac{\pi}{2}-\frac{1}{2}\cot^{-1}\biggl(\frac{-N+6}{2\sqrt{N-1}}\biggr)$$
    $$=\frac{\pi}{2}-\frac{1}{2}\tan^{-1}\biggl(\frac{2\sqrt{N-1}}{-N+6}\biggr)$$
    $$=\frac{1}{2}\biggl(\tan^{-1}(0)-\tan^{-1}\biggl(\frac{2\sqrt{N-1}}{-N+6}\biggr)\biggr).$$

    Using mean value theorem again, we know that there exists a $\tilde{\theta}_N$ in $(0,\frac{2\sqrt{N-1}}{-N+6})$ such that

    $$|R_2\cup R_3| = \frac{1}{2}\frac{1}{1+(\tilde{\theta}_N)^2}\biggl(0-\frac{2\sqrt{N-1}}{-N+6}\biggr)$$
    $$=\frac{\sqrt{N-1}}{N-6}\frac{1}{1+(\tilde{\theta}_N)^2}$$
    $$\leq \frac{\sqrt{N-1}}{N-6},$$

    \noindent since $1+(\tilde{\theta}_N)^2\geq1$. Hence $|R_2\cup R_3|$ is of order $O(1/\sqrt{N})$. Using this with what was just proved about $|R_2|$, we have that $\frac{|R_2|}{|R_2\cup R_3|}$ is of order $O(1/N)$.
\end{proof}

\begin{remark}
    First we will call attention to the fact that each iteration of Grover's algorithm is approximately a rotation of the state vector by $2/\sqrt{N}$ for $N$ large enough. This is a common fact in the foundational literature, and we point the reader to \cite{BBHT} for its derivation. While with this analysis we cannot definitively say that a locally optimized Grover's algorithm will always bypass $R_2$ and get trapped in $R_3$, we can say that it is highly likely for this to happen. This is due to the size of $R_2\cup R_3$ being of order $O(1/\sqrt{N})$ and the proportion of that which is $R_2$ being of order $O(1/N)$.
\end{remark}

\section{Conclusion \& Further Research}\label{sec6}

We have found that in the setting of generalized Grover's algorithm allowing a general phase change at each step of iteration, with a sufficiently large system, and an initialization set as the Hadamard gate, that the $\alpha$ domain is partitioned into three distinct regions in which the qualitatively different optimal phase change is determined. The regions are precisely identified and the corresponding numerical results for concrete cases were presented to demonstrate scenarios in which the greedy search algorithm leads to non-exact search. However, if the initial amplitude vector is complex, and we know how to describe it, then the optimal phase change depends non-trivially on the phase of the target state's amplitude and it can be found by solving the equation \eqref{eq:Optimizaiton} at each iteration. 

Moving forward there are a few areas of research we plan to investigate further. The first is identifying scenarios in which initializing with a non-Hadamard gate could become more beneficial for possible speed-up, and if those lead to amplitude vectors that are complex. If this occurs, then a generalized Grover's algorithm with a phase change that differs from $\pi$ and that non-trivially depends on the phase of the amplitude vector would outperform $\phi$ = $\pi$ (The phase change used in the classical Grover's iterate) as seen in Figure \ref{fig:2}. A non-Hadamard initial gate would imply that there is some underlying knowledge of the data set that we can take advantage of so as to choose a specific class of probability distributions in the beginning, rather than a uniform one (as in Hadamard's gate).

The second applies a similar analysis to the topic of phase matching in deterministic Grover's algorithm. We aim to identify distinct regions where the phase matching criteria deviate from $\phi=\psi$, the surface where two phases present in Grover's match, while also deriving an optimization formula to determine the corresponding phase matching.

The third is better understanding of quantum random walks. Justified by numerous applications, classical random walks are an area of interest to us, and naturally we wish to understand their quantum counterparts, both as tools to search spatial data, and as stochastic objects in and of themselves. We look forward to pursue the above research directions to understand both unstructured and structured search algorithms arising in quantum computation.

\bibliography{sn-bibliography}

\end{document}